\RequirePackage{fix-cm}
\documentclass[smallextended]{svjour3}       
\smartqed  
\usepackage{titlesec}
\usepackage{latexsym}
\usepackage{mathrsfs}
\usepackage{multirow}

\usepackage{amsfonts,amssymb,amsmath,amsthm,bm}
\usepackage{color}
\usepackage{graphics}
\usepackage{makecell}
\usepackage{float}
\usepackage{autobreak} 
\usepackage{booktabs}

\usepackage[figuresright]{rotating}
\usepackage{array}
\allowdisplaybreaks
\usepackage{siunitx}

\def\BF{\mathbb{F}}

\def\bv{\mathbf{v}}
\def\bu{\mathbf{u}}
\def\bH{\mathbf{H}}
\def\bO{\mathbf{O}}
\def\bI{\mathbf{I}}
\def\bA{\mathbf{A}}
\def\bP{\mathbf{P}}
\def\bQ{\mathbf{Q}}
\def\bM{\mathbf{M}}

\def\bB{\mathbf{B}}
\def\bV{\mathbf{V}}

\def\bU{\mathbf{U}}

\def\cF{\mathcal{F}}
\def\cD{\mathcal{D}}
\def\cS{\mathcal{S}}
\def\cC{\mathcal{C}}

\def\cH{\mathcal{H}}
\def\cP{\mathcal{P}}
\def\cQ{\mathcal{Q}}
\def\cX{\mathcal{X}}

\def\cX{\mathcal{X}}

\def\cF{\mathcal{F}}
\def\cG{\mathcal{G}}

\def\tcC{\Tilde{\mathcal{C}}}

\def\tcS{\tilde{\mathcal{S}}}
\def\tcF{\Tilde{\mathcal{F}}}
\def\tbv{\Tilde{\mathbf{v}}}
\def\tbu{\Tilde{\mathbf{u}}} 

\def\obv{\bar{\mathbf{v}}}
\def\obu{\bar{\mathbf{u}}}

\def\sL{\mathscr{L}}
\def\rank{\mathrm{rank}}
\def\rs{\mathrm{rs}}
\def\wt{\mathrm{wt}}
\def\RREF{\mathrm{RREF}}
\def\EF{\mathrm{EF}}

\titleformat{\subsection}
{\normalfont\large\bfseries}
{\bfseries\thesubsection}
{1em}
{}

\newtheorem{construction}{Construction}
\journalname{}

\begin{document}
	
	\title{Quasi-MSRD Codes and Their Properties \thanks{}
	}
	
	
	\author{Qingfeng Xia  \and
		Fang-Wei Fu
	}
	

	\institute{Qingfeng Xia\at
		Chern Institute of Mathematics and LPMC, Nankai University, Tianjin, 300071, China\\
		\email{1120250016@mail.nankai.edu.cn}             \\
		\and
		Fang-Wei Fu\at
		Chern Institute of Mathematics and LPMC, Nankai University, Tianjin, 300071, China\\
		\email{fwfu@nankai.edu.cn}}
	\date{Received: date / Accepted: date}

	\maketitle
	\begin{abstract}
 Sum-rank-metric codes have recently attracted considerable attention of many researchers, due to their applications in network coding, space-time codes and distributed storage. MSRD codes are those good codes attaining the Singleton bound in the sum-rank metric. However, MSRD codes do not exist for some dimensions. Motivated by this fact, we introduce the notion of quasi-MSRD (QMSRD) codes and provide some properties of them. What is more, we find that not every QMSRD code has a QMSRD dual code, so we give the definition of dually QMSRD codes, whose dual codes and themselves are both QMSRD. Finally, we characterize such codes, derive their support, rank-list and sum-rank distributions as well as compute their generalized sum-rank weights.
	\keywords{Sum-rank metric \and QMSRD codes \and Dually QMSRD codes \and Weight distributions \and Generalized weights}
	\end{abstract}

\section{Introduction}
Sum-rank-metric codes were introduced in \cite{NU} for their application in linear network coding, which can be used to detect and correct errors in multiplicative-additive finite-field matrix channels as rank-metric codes.
While rank-metric codes are employed to obtaining one-shot matrix codes, which use the matrix channel only once, sum-rank-metric codes are considered for consecutive uses of matrix channels that can not only reduce the size of network alphabet but also detect and correct more errors.
In addition, sum-rank-metric codes were also widely applied in space-time coding, see \cite{LK,SK}, and coding for distributed storage, see \cite{CMST,MPK}. 

In recent years, there has been tremendous interest in developing the theory of sum-rank-metric codes. Some scholars did research on the construction of sum-rank-metric codes, such as \cite{BCa,BC,C,CCQ,LCJXZ,Mar,Ma,M,MZ}, while others considered various bounds of them, such as \cite{AGKP,AKR,BGR,LCJXZ}. 
Besides these works, some researchers also investigated decoding algorithms of sum-rank-metric codes, see \cite{BJPR,CCQ,WCZL}. 

In most of early works, the authors only considered those sum-rank-metric codes, where the number of rows are equal at different blocks. 
The general Singleton bound of sum-rank-metric codes for arbitrary numbers of rows was given in \cite{BGR}. 
Those sum-rank-metric codes attaining the Singleton bound are called maximum sum-rank distance (MSRD) codes.
In \cite{BGR}, the block lengths of MSRD codes were upper bounded.
The authors of \cite{BGR} also defined the dual code of a sum-rank-metric code and the sum-rank, rank-list and support distributions. 
They derived the MacWilliams identities of rank-list and support distributions and based on these results, computed the sum-rank, rank-list and support distributions of an MSRD code.
In \cite{BGRa}, anticodes in the sum-rank metric were introduced. 
Right after that, the authors modified the anticode bound in \cite{BGRa} and gave the definition of optimal anticodes in \cite{CGLGMS}. 
Using the notion of optimal anticodes, the authors defined the generalized sum-rank weights.
Based on the knowledge about optimal anticodes in the sum-rank metric, they also studied some properties related to being MSRD codes and computed the generalized sum-rank weights of them.
Besides these works, there are many other works which have to do with MSRD codes, see \cite{C,LWWS,Mar,MP,Marti,Mart,MSK,SS}. 

Rank-metric codes can be regarded as a special kind of sum-rank-metric codes, which have only one block. Rank-metric codes achieving the Singleton bound are called maximum rank distance (MRD) codes. MRD codes are good codes, which have the largest possible minimum rank distance, but they only exist if the number of rows divides the dimension of the code.  
In \cite{CGLR}, the authors defined quasi-MRD (QMRD) codes as codes which have the largest minimum rank distance for the fixed length and dimension, but are not MRD codes. 
They considered the dual codes of QMRD codes. If the dual code of a QMRD code is a QMRD code as well, then they called it the dually QMRD code. 
They offered a sufficient and necessary condition for a QMRD code to be a dually QMRD code. 
What is more, they discussed the generalized weights of dually QMRD codes and computed the rank distribution of them.

Apart from this, linear codes in the folded Hamming distance can be seen as another special kind of sum-rank-metric codes. The authors of \cite{MR} introduced quasi-MDS (QMDS) codes as an alternative to MDS codes in the folded Hamming distance for dimensions for which MDS codes in the folded Hamming distance do not exist. Analogously, they also defined dually QMDS codes and gave the weight distribution of them. 

We find that MSRD codes do not exist for some dimensions just like MRD codes and MDS codes in the folded Hamming metric. 
Thus, inspired by those related works mentioned above, we introduce the notion of quasi-MSRD (QMSRD) codes and dually QMSRD codes, whose dual codes and themselves are both QMSRD codes. 
QMSRD codes have the largest possible minimum sum-rank distances, thus offer good capabilities for error detection and error correction.
Some of them are subspaces of MSRD codes, but some others do not.
Therefore, it is significant to do research on their properties.
In this work, our main contributions are as follows:
\begin{itemize}
  \item We provide some properties of QMSRD codes, see \textbf{Propositions \ref{Proposition 3.5}, \ref{Proposition 3.7} and Theorems \ref{Theorem 3.8}, \ref{Theorem 3.11}}.
  \item We characterize the dually QMSRD codes in terms of the minimum sum-rank distance of their dual codes and themselves as well as the number of their minimum sum-rank weight codewords, see \textbf{Proposition \ref{Proposition 4.6} and Theorem \ref{Theorem 4.14}}.
  \item We give the support, rank-list and sum-rank distributions of dually QMSRD codes, see \textbf{Theorem \ref{Theorem 4.18} and Corollaries \ref{Corollary 4.19}, \ref{Corollary 4.20}}.
  \item We compute the generalized sum-rank weights of dually QMSRD codes, see \textbf{Theorem \ref{Theorem 4.27}}.
\end{itemize}

This paper is organized as follows. 
In Section 2, we start by giving necessary notations. Then we review some basic definitions and bounds of sum-rank-metric codes which will be used later. 
In Section 3, we define QMSRD codes and provide some properties which are closely related to being QMSRD codes. In addition, we give an upper bound of the block lengths of QMSRD codes and find that some punctured codes of QMSRD codes are QMSRD codes as well. 
In Section 4, we discuss the dual codes of QMSRD codes and introduce the notion of dually QMSRD codes. Next, we try to give the sufficient and necessary condition for a QMSRD code to be a dually QSMRD code. Finally, we compute the support, rank-list and sum-rank distributions as well as the generalized sum-rank weights of dually QMSRD codes.
In Section 5, we conclude this paper.

\section{Preliminaries}
In this preliminary section, we firstly give some notations.
Throughout this paper, $\mathbb{F}_{q}$ is the finite field with $q$ elements, where $q$ is a power of a prime.
Also, we let $\mathbb{N}=\{1,2,\ldots\}$ and $\mathbb{N}_{0}=\mathbb{N}\cup\{0\}$. 
For $a\in\mathbb{N}$, we denote the set $\{1,\ldots,a\}$ by $[a]$.
Let $a$ and $b$ be integers, then the $q$-ary Gaussian coefficient of $a$ and $b$ is given as below
$$\binom{a}{b}_{q}=\begin{cases}
                     0, & \mbox{if } a<0,\,b<0,\,b>a, \\
                     1, & \mbox{if } b=0 \text{ and } a\geq 0, \\
                     \frac{(q^{a}-1)(q^{a-1}-1)\cdots(q^{a-b+1}-1)}{(q^{b}-1)(q^{b-1}-1)\cdots(q-1)}, & \mbox{otherwise}.
                   \end{cases}$$

Let $\boldsymbol{x}=(x_{1},\ldots,x_{n})$, $\boldsymbol{y}=(y_{1},\ldots,y_{n})\in\mathbb{F}_{q}^{n}$, the Hamming distance between $\boldsymbol{x}$ and $\boldsymbol{y}$ is given as follows
$$d_{H}(\boldsymbol{x},\boldsymbol{y})=|\{i\,:\,x_{i}\neq y_{i}, 1\leq i\leq n\}|.$$
Let $C\subseteq\mathbb{F}_{q}^{n}$ be a code, its minimum Hamming distance is defined as follows 
$$d_{H}(C)=\min\{d_{H}(\boldsymbol{c}_{1},\boldsymbol{c}_{2})\,:\,\boldsymbol{c}_{1}\neq\boldsymbol{c}_{2}\in C\}.$$
If $C$ is a linear subspace of $\mathbb{F}_{q}^{n}$, $C$ is called a linear code. 

Next, we revisit some basic definitions and known results about sum-rank-metric codes, most of which can be found in \cite{BGR} and \cite{CGLGMS}. 
We let $m_{1},\ldots,m_{t}$, $n_{1},\ldots,n_{t}$ be $2t$ positive integers such that $m_{i}\geq n_{i}$ for all $i\in[t]$, $m_{1}\geq\cdots\geq m_{t}$ and $n=n_{1}+\cdots +n_{t}$. 
Then we define the linear space as follows:  
$$\mathbb{F}_{q}^{(m_{1},n_{1}),\ldots,(m_{t},n_{t})}=\mathbb{F}_{q}^{(m_{1},n_{1})}\oplus\cdots\oplus\mathbb{F}_{q}^{(m_{t},n_{t})},$$
where $\mathbb{F}_{q}^{(m_{i},n_{i})}$ is the set consisting of all $m_{i}\times n_{i}$ matrices over $\mathbb{F}_{q}$.
Obviously, this is a linear space over $\mathbb{F}_{q}$ with dimension $\sum\limits_{i=1}^{t}m_{i}n_{i}$.

\begin{definition}
  {\rm Let $M=(M_{1},\ldots,M_{t})\in\mathbb{F}_{q}^{(m_{1},n_{1}),\ldots,(m_{t},n_{t})}$, the \emph{sum-rank weight} of $M$ is defined as follows $$wt_{sr}(M)=\sum_{i=1}^{t}{\rm rank}(M_{i}),$$
  and for $M=(M_{1},\ldots,M_{t})$, $N=(N_{1},\ldots,N_{t})\in\mathbb{F}_{q}^{(m_{1},n_{1}),\ldots,(m_{t},n_{t})}$, the \emph{sum-rank distance} between $M$ and $N$ is given as below
  $$d_{sr}(M,N)=wt_{sr}(M-N).$$}
\end{definition}
It is easy to prove that this is a metric on $\mathbb{F}_{q}^{(m_{1},n_{1}),\ldots,(m_{t},n_{t})}$.
\begin{definition}
  {\rm A subset of $\mathbb{F}_{q}^{(m_{1},n_{1}),\ldots,(m_{t},n_{t})}$ endowed with the sum-rank metric is called a \emph{sum-rank-metric code} with block length $t$ and matrix sizes $m_{1}\times n_{1},\ldots,m_{t}\times n_{t}$. The \emph{minimum sum-rank weight and minimum sum-rank distance} of $C$ are defined as usual via
  $$wt_{sr}(C)=\min\{wt_{sr}(M)\,:\,\boldsymbol{0}\neq M\in C\}$$ and $$d_{sr}(C)=\min\{d_{sr}(M,N)\,:\,M\neq N\in C\}.$$}
\end{definition}
When $C=\{\boldsymbol{0}\}$ or $C=\mathbb{F}_{q}^{(m_{1},n_{1}),\ldots,(m_{t},n_{t})}$, we say that the sum-rank-metric code $C\subseteq\mathbb{F}_{q}^{(m_{1},n_{1}),\ldots,(m_{t},n_{t})}$ is trivial. Otherwise, we say it is non-trivial.
And when $C$ is a linear subspace of $\mathbb{F}_{q}^{(m_{1},n_{1}),\ldots,(m_{t},n_{t})}$, $C\subseteq\mathbb{F}_{q}^{(m_{1},n_{1}),\ldots,(m_{t},n_{t})}$ is called a linear sum-rank-metric code.
\begin{remark}
  {\rm When $m_{1}=\cdots=m_{t}=n_{1}=\cdots=n_{t}=1$, the sum-rank metric reduces to the Hamming metric and when $t=1$, it reduces to the rank-metric. What is more, when $m_{1}=\cdots=m_{t}\geq 2$ and $n_{1}=\cdots=n_{t}=1$, it can be seen as the folded Hamming metric. For more details about the folded Hamming metric, the readers can refer to \cite{MR}.}
\end{remark}
 
 We have an $\mathbb{F}_{q}$-linear isomorphism: $\mathbb{F}_{q}^{(m_{1},n_{1}),\ldots,(m_{t},n_{t})}\cong \mathbb{F}_{q^{m_{1}}}^{n_{1}}\oplus\cdots\oplus\mathbb{F}_{q^{m_{t}}}^{n_{t}}.$
 When $m_{1}=\cdots=m_{t}=m$, we let $n=n_{1}+\cdots+n_{t}$. 
  Then $\mathbb{F}_{q}^{(m_{1},n_{1}),\ldots,(m_{t},n_{t})}\cong\mathbb{F}_{q^m}^{n}$. 
  We can give the isomorphism exactly.
  Let $\mathcal{B}=\{\gamma_{1},\ldots,\gamma_{m}\}$ be the basis of $\mathbb{F}_{q^m}$ over $\mathbb{F}_{q}$ and $\boldsymbol{x}=(\boldsymbol{x}^{(1)},\ldots,\boldsymbol{x}^{(t)})\in\mathbb{F}_{q^m}^{n}$, where $\boldsymbol{x}^{(i)}=(x^{(i)}_{1},\ldots,x^{(i)}_{n_{i}})\in\mathbb{F}_{q^m}^{n_i}$.
  Set $Mat_{\mathcal{B}}^{(i)}\,:\,\mathbb{F}_{q^m}^{n_i}\rightarrow\mathbb{F}_{q}^{(m,n_i)},\,\boldsymbol{x}^{(i)}\mapsto M_{\mathcal{B}}^{(i)}$,
  where $x^{(i)}_{s}=\sum\limits_{j=1}^{m}(M_{\mathcal{B}}^{(i)})_{js}\gamma_{j}$.
  In other words, the $s$-th column of $M_{\mathcal{B}}^{(i)}$ is the expansion of $x^{(i)}_{s}$ over the basis $\mathcal{B}$.
  Set $Mat_{\mathcal{B}}\,:\,\mathbb{F}_{q^m}^{n}\rightarrow\mathbb{F}_{q}^{(m_{1},n_{1}),\ldots,(m_{t},n_{t})},\,\boldsymbol{x}\mapsto M_{\mathcal{B}}$, 
  where $M_{\mathcal{B}}=(M_{\mathcal{B}}^{(1)},\ldots,M_{\mathcal{B}}^{(t)})
  =(Mat_{\mathcal{B}}^{(1)}(\boldsymbol{x}^{(1)}),\ldots,Mat_{\mathcal{B}}^{(t)}(\boldsymbol{x}^{(t)}))$.
  
  This is an $\mathbb{F}_{q}$-linear isomorphism between these two linear spaces over $\mathbb{F}_{q}$.
  $C$ is an $\mathbb{F}_{q}$-linear code in $\mathbb{F}_{q^m}^{n}$ if and only if $M_{\mathcal{B}}(C)$ is a linear sum-rank-metric code in $\mathbb{F}_{q}^{(m_{1},n_{1}),\ldots,(m_{t},n_{t})}$.
  Thus, if $C$ is a linear code in $\mathbb{F}_{q^m}^{n}$, $M_{\mathcal{B}}(C)$ is a linear sum-rank-metric code in $\mathbb{F}_{q}^{(m_{1},n_{1}),\ldots,(m_{t},n_{t})}$.
  However, we cannot guarantee the linearity of $C$ over $\mathbb{F}_{q^m}$ even when $M_{\mathcal{B}}(C)$ is a linear sum-rank-metric code in $\mathbb{F}_{q}^{(m_{1},n_{1}),\ldots,(m_{t},n_{t})}$.
  Note that the isomorphism will be used later.
  
  In the following, we turn to the dual code of sum-rank-metric codes. We let $M=(M_{1},\ldots,M_{t})$, $N=(N_{1},\ldots,N_{t})\in\mathbb{F}_{q}^{(m_{1},n_{1}),\ldots,(m_{t},n_{t})}$ 
  and denote by $Tr(A)$ the sum of the elements on the main diagonal of the square matrix $A$.
  Then the trace inner product of $M$ and $N$ is defined as 
  $$\langle M,N \rangle_{tr}=\sum\limits_{i=1}^{t}Tr(M_{i}N_{i}^{\top}).$$
  \begin{definition}
    {\rm Let $C$ be a linear sum-rank-metric code in $\mathbb{F}_{q}^{(m_{1},n_{1}),\ldots,(m_{t},n_{t})}$. Then the \emph{dual code} of $C$ is 
    $$C^{\bot_{tr}}=\{N\in\mathbb{F}_{q}^{(m_{1},n_{1}),\ldots,(m_{t},n_{t})}\,:\,\langle M,N\rangle_{tr}=0,\,\forall M\in C\}.$$}
  \end{definition}
  
  In the following, we review several important bounds on the size of sum-rank-metric codes. We firstly present the Singleton bound for linear sum-rank-metric codes.
  \begin{theorem}(\!\cite[Theorem III.2]{BGR})\label{Theorem 2.5}
  Let $C\subseteq\mathbb{F}_{q}^{(m_{1},n_{1}),\ldots,(m_{t},n_{t})}$ be a nontrivial linear sum-rank-metric code with $d_{sr}(C)=d$. 
  Let $j\in[t]$ and $0\leq\delta\leq n_{j}-1$ be the unique integers satisfying $d-1=\sum\limits_{i=1}^{j-1}n_{i}+\delta$.
  Then $$\dim(C)\leq \sum\limits_{i=j}^{t}m_{i}n_{i}-m_{j}\delta.$$
  In the special case where $m_{1}=\cdots=m_{t}=m$, the upper bound simplifies to $\dim(C)\leq m(n-d+1)$.
  \end{theorem}
  If a linear sum-rank-metric code attains this bound, we call it the maximum sum-rank distance (MSRD) code.
  For convenience, it is also necessary to recall another Singleton-type bound, although it is weaker than the bound above.
  \begin{proposition}(\!\cite[Corollary VI.7]{CGLGMS})\label{Proposition 2.6}
  Let $C\subseteq\mathbb{F}_{q}^{(m_{1},n_{1}),\ldots,(m_{t},n_{t})}$ be a nontrivial linear sum-rank-metric code. Let $j\in[t]$, $0\leq\delta\leq n_{j}-1$ and 
  $0\leq s\leq m_{j}-1$ be the unique integers satisfying $\dim(C)=\sum\limits_{i=1}^{j-1}m_{i}n_{i}+\delta m_{j}+s$.
  Then $$d_{sr}(C)\leq \sum\limits_{i=j}^{t}n_{i}-\delta+\begin{cases}
                                                           1, & \mbox{if } s=0, \\
                                                           0, & \mbox{else}.
                                                         \end{cases}$$
  \end{proposition}
  The following theorem is another bound called projective sphere-packing bound which will be used to get an upper bound on the block length of QMSRD codes. For $r\in\mathbb{N}$, we let $V_{r}$ be the volume of any sphere in $\mathbb{F}_{q}^{(m_{1},n_{1}),\ldots,(m_{t},n_{t})}$ of sum-rank radius $r$. Then $$V_{r}=\sum\limits_{s=0}^{r}\,\sum\limits_{\substack{(s_{1},\ldots,s_{t})\in\mathbb{N}_{0}^{t}\\s_{1}+\cdots+s_{t}=s}}
  \prod\limits_{i=1}^{t}\binom{n_{i}}{s_{i}}_{q}\prod\limits_{j=0}^{s_{i}-1}(q^{m_{i}}-q^{j}).$$
  \begin{theorem}(\!\cite[Theorem III.7]{BGR})\label{Theorem 2.7}
 Let $\{\boldsymbol{0}\}\neq C\subseteq\mathbb{F}_{q}^{(m_{1},n_{1}),\ldots,(m_{t},n_{t})}$ be a sum-rank-metric code with minimum sum-rank distance $3\leq d_{sr}(C)\leq n$. 
 Let $l\in[t-1]$ and $0\leq\delta\leq n_{l+1}-1$ be the unique integers such that $d-3=\sum\limits_{j=1}^{l}n_{j}+\delta$ and $V_{1}^{\prime}$ be the volume of any sphere in $\mathbb{F}_{q}^{(m_{l+1},n_{l+1}-\delta),(m_{l+2},n_{l+2}),\ldots,(m_{t},n_{t})}$ of sum-rank radius $1$. Then
 $$|C|\leq\left\lfloor\frac{q^{\sum\limits_{i=l+1}^{t}m_{i}n_{i}-\delta m_{l+1}}}{V_{1}^{\prime}}\right\rfloor,$$
 where $$V_{1}^{\prime}=1+\frac{q^{n_{l+1}-\delta}-1}{q-1}(q^{m_{l+1}}-1)+\sum\limits_{i=l+2}^{t}\frac{q^{n_{i}}-1}{q-1}(q^{m_{i}}-1).$$
  \end{theorem}
  We end up this section with the definition of generalized sum-rank weights. In order to make it clear, we firstly give the anticode bound for sum-rank-metric codes. Note that ${\rm maxsrk}(C)=\max\{d_{sr}(M,N)\,:\,M\neq N\in C\}.$
  \begin{theorem}(\!\cite[Theorem IV.1]{CGLGMS})
  Let $C\subseteq\mathbb{F}_{q}^{(m_{1},n_{1}),\ldots,(m_{t},n_{t})}$ be a linear sum-rank-metric code. Then 
  $$\dim(C)\leq\max\left\{\sum\limits_{i=1}^{t}m_{i}{\rm rank}(M_{i})\,:\,M=(M_{1},\ldots,M_{t})\in C\right\}.$$
  In particular, if $m_{1}=\cdots=m_{t}=m$, then
  $$\dim(C)\leq m{\rm maxsrk}(C).$$
  \end{theorem}
  Optimal anticodes are those sum-rank-metric codes which meet the anticode bound. And when $t=1$, the result reduces to the anticode bound for rank-metric codes.
  \begin{definition}
   {\rm Let $C\subseteq\mathbb{F}_{q}^{(m_{1},n_{1}),\ldots,(m_{t},n_{t})}$ be a linear sum-rank-metric code. For each $r\in[\dim(C)]$, the \emph{$r$-th generalized sum-rank weight} of $C$ is defined as
   \begin{equation*}
   \begin{aligned}
   d_r(C) = & \min \{{\rm maxsrk}(\mathcal{A})\,:\,\mathcal{A}=\mathcal{A}_1 \times \cdots \times \mathcal{A}_{t} \text{ where } \mathcal{A}_i \subseteq \mathbb{F}_q^{m_i \times n_i} \text{ is an }\\ & \text{optimal anticode and } \dim(C \cap \mathcal{A}) \geq r\}.
   \end{aligned}
   \end{equation*}
   In particular, if $m_{1}=\cdots=m_{t}=m$, then
   \begin{equation*}
   \begin{aligned}
   d_r(C) = &\frac{1}{m}\min \{\dim(\mathcal{A})\,:\,\mathcal{A}=\mathcal{A}_1 \times \cdots \times \mathcal{A}_{t} \text{ where } \mathcal{A}_i \subseteq \mathbb{F}_q^{m_i \times n_i} \text{ is an }\\ & \text{optimal anticode and } \dim(C \cap \mathcal{A}) \geq r\}.
   \end{aligned}
   \end{equation*}}
  \end{definition}

\section{Quasi-MSRD Codes}
In this section, we give the definition of quasi-MSRD (QMSRD) codes and derive some properties of them.
According to the preliminaries above, a linear sum-rank-metric code $C\subseteq\mathbb{F}_{q}^{(m_{1},n_{1}),\ldots,(m_{t},n_{t})}$ is MSRD if there exist integers $j\in[t]$ and $0\leq\delta\leq n_{j}-1$ such that
$$\dim(C)=\sum\limits_{i=j}^{t}m_{i}n_{i}-m_{j}\delta \text{ and } d_{sr}(C)=\sum\limits_{i=1}^{j-1}n_{i}+\delta+1.$$ 
This is a class of optimal sum-rank-metric codes, but MSRD codes only exist if the dimension $\dim(C)$ can be written as the form $\sum\limits_{i=j}^{t}m_{i}n_{i}-m_{j}\delta$. 
If not, we let the dimension $\dim(C)=\sum\limits_{i=j}^{t}m_{i}n_{i}-m_{j}\delta-s$, where $j$, $\delta$ and $s$ are integers such that $j\in[t]$, $0\leq\delta\leq n_{j}-1$ and $0<s\leq m_{j}-1$. By the Singleton bound, $d_{sr}(C)\leq\sum\limits_{i=1}^{j-1}n_{i}+\delta+1$. 
Motivated by this, we define quasi-MSRD (QMSRD) codes as below.
\begin{definition}
  {\rm Let $C\subseteq\mathbb{F}_{q}^{(m_{1},n_{1}),\ldots,(m_{t},n_{t})}$ be a linear sum-rank-metric code. $C$ is called a \emph{quasi-MSRD (QMSRD)} code if there exist integers $j\in[t]$, $0\leq\delta\leq n_{j}-1$ and $0<s\leq m_{j}-1$ such that 
  $$\dim(C)=\sum\limits_{i=j}^{t}m_{i}n_{i}-m_{j}\delta-s \text{ and } d_{sr}(C)=\sum\limits_{i=1}^{j-1}n_{i}+\delta+1.$$}
\end{definition}
\begin{remark}
  {\rm When $t=1$, quasi-MSRD codes reduce to quasi-MRD codes and when $m_{1}=\cdots=m_{t}\geq 2$ and $n_{1}=\cdots=n_{t}=1$, quasi-MSRD codes can be seen as quasi-MDS codes. For more details about these two kinds of codes, readers can refer to \cite{CGLR} and \cite{MR}.
  
  When $m_{1}=\cdots=m_{t}=m$ and $n=n_{1}+\cdots+n_{t}$, let $C\subseteq\mathbb{F}_{q}^{(m,n_{1}),\ldots,(m,n_{t})}$ be the linear sum-rank-metric code.
  Then $\dim(C)\leq m(n-d_{sr}(C)+1)$ by the Singleton bound, giving that $d_{sr}(C)\leq n-\left\lceil\frac{\dim(C)}{m}\right\rceil+1$.
  At this special case, MSRD codes only exist when $\dim(C)$ is divisible by $m$. 
  More precisely, if $\dim(C)$ is divisible by $m$, $\dim(C)=m(n-d_{sr}(C)+1)$ if and only if $d_{sr}(C)=n-\left\lceil\frac{\dim(C)}{m}\right\rceil+1=n-\frac{\dim(C)}{m}+1$. Under such conditions, $C$ is an MSRD code.
  If $\dim(C)$ is not divisible by $m$, $\dim(C)=m(n-d_{sr}(C)+1)-s$, $0<s\leq m-1$ if and only if $d_{sr}(C)=n-\left\lceil\frac{\dim(C)}{m}\right\rceil+1$, which implies that $C$ is a QMSRD code.}
\end{remark}
It is easy to prove that some linear subspaces of MSRD codes are QMSRD codes.
\begin{proposition}\label{Proposition 3.3}
  Let $j$, $\delta$ and $s$ be integers such that $j\in[t]$, $0\leq\delta\leq n_{j}-1$ and $0<s\leq m_{j}-1$ and 
  $C\subseteq\mathbb{F}_{q}^{(m_{1},n_{1}),\ldots,(m_{t},n_{t})}$ be an MSRD code with minimum distance $d_{sr}(C)=\sum\limits_{i=1}^{j-1}n_{i}+\delta+1$ and dimension $\dim(C)=\sum\limits_{i=j}^{t}m_{i}n_{i}-m_{j}\delta$.
  If $D$ is a linear subspace of $C$ with dimension $\dim(D)=\sum\limits_{i=j}^{t}m_{i}n_{i}-m_{j}\delta-s$, then $d_{sr}(D)=d_{sr}(C)$ and $D$ is a QMSRD code.
\end{proposition}
\begin{proof}
  Since $D$ is a linear subspace of $C$, $d_{sr}(D)\geq d_{sr}(C)$. 
  Suppose that $d_{sr}(D)>d_{sr}(C)$, i.e., $d_{sr}(D)\geq \sum\limits_{i=1}^{j-1}n_{i}+\delta+2$, 
  then by the Singleton bound, $\dim(D)\leq\sum\limits_{i=j}^{t}m_{i}n_{i}-m_{j}(\delta+1)$, which contradicts the fact that $\dim(D)=\sum\limits_{i=j}^{t}m_{i}n_{i}-m_{j}\delta-s$ and $0<s\leq m_{j}-1$. 
  Hence, $d_{sr}(D)=d_{sr}(C)$ and $D$ is a QMSRD code.
\end{proof}
However, not every QMSRD code arises as a linear subspace of an MSRD code. The following is an example.
\begin{example}\label{Example 3.4}
  Let $C\subseteq\mathbb{F}_{2}^{(2,2),(2,2)}$ be a linear sum-rank-metric code generated by the following elements: 
  $$\left(\begin{pmatrix}
            1 & 0 \\
            0 & 1 
          \end{pmatrix},\,\begin{pmatrix}
            1 & 0 \\
            0 & 0 
          \end{pmatrix}\right),\,
          \left(\begin{pmatrix}
            0 & 0 \\
            1 & 0 
          \end{pmatrix},\,\begin{pmatrix}
            0 & 1 \\
            1 & 0 
          \end{pmatrix}\right),\,
          \left(\begin{pmatrix}
            0 & 1 \\
            0 & 0 
          \end{pmatrix},\,\begin{pmatrix}
            1 & 0 \\
            0 & 1 
          \end{pmatrix}\right).$$ 
  $C$ has dimension $\dim(C)=3$ and minimum sum-rank distance $d_{sr}(C)=3$. Therefore, $C$ is a QMSRD code. 
  However, according to \cite[Example VI.9]{BGR}, there exists no MSRD code with minimum sum-rank distance $3$ in $\mathbb{F}_{2}^{(2,2),(2,2)}$. 
  Hence, $C$ cannot be a linear subspace of an MSRD code as Proposition \ref{Proposition 3.3}.
  In addition, we find that $\dim(C^{\bot_{tr}})=5$ and $d_{sr}(C^{\bot_{tr}})=2$, so $C^{\bot_{tr}}$ is a QMSRD code as well. 
\end{example}
Next, we study some properties which are closely related to being QMSRD codes. Firstly, using optimal anticodes in the sum-rank metric, we derive the result below.
\begin{proposition}\label{Proposition 3.5}
  Let $C$ be a linear sum-rank-metric code in $\mathbb{F}_{q}^{(m_{1},n_{1}),\ldots,(m_{t},n_{t})}$. Then the following are equivalent:
   \begin{itemize}
    \item[(1)] $C$ is a QMSRD code.
    \item[(2)] There exist integers $j\in[t]$, $0\leq\delta\leq n_{j}-1$, $0<s\leq m_{j}-1$ such that 
    $\dim(C)=\sum_{i=j}^{t}m_{i}n_{i}-m_{j}\delta-s$ and for any optimal anticode $\mathcal{A}$ with ${\rm maxsrk}(\mathcal{A})\leq \sum\limits_{i=1}^{j-1}n_{i}+\delta$ one has $C\cap\mathcal{A}=\{\boldsymbol{0}\}$.
  \end{itemize}
\end{proposition}
\begin{proof}
  If $C$ is QMSRD, then $\dim(C)=\sum_{i=j}^{t}m_{i}n_{i}-m_{j}\delta-s$ and $d_{sr}(C)=\sum\limits_{i=1}^{j-1}n_{i}+\delta+1$, where $j$, $\delta$ and $s$ are integers such that $j\in[t]$, $0\leq\delta\leq n_{j}-1$ and $0<s\leq m_{j}-1$. 
  Let $\mathcal{A}$ be an optimal anticode with ${\rm maxsrk}(\mathcal{A})\leq \sum\limits_{i=1}^{j-1}n_{i}+\delta=d_{sr}(C)-1$. 
  For every $\boldsymbol{0}\neq M\in C$, one has 
  $wt_{sr}(M)\geq d_{sr}(C)>{\rm maxsrk}(\mathcal{A})$, which implies that $M\notin\mathcal{A}$.
  Thus, $C\cap\mathcal{A}=\{\boldsymbol{0}\}$.
  
  Suppose $C$ satisfies $(2)$, by the Singleton bound, $d_{sr}(C)\leq\sum\limits_{i=1}^{j-1}n_{i}+\delta+1$. Let $M=(M_{1},\ldots,M_{t})\in C$. 
  For each $i\in[t]$, there is an optimal rank-metric anticode $\mathcal{A}_{i}\subseteq\mathbb{F}_{q}^{(m_{i},n_{i})}$ with $\dim(\mathcal{A}_{i})=m_{i}{\rm rank}(M_{i})$ which contains $M_{i}$. 
  Therefore, $\mathcal{A}=\mathcal{A}_{1}\times\cdots\times\mathcal{A}_{t}$ is an optimal anticode with ${\rm maxsrk}(\mathcal{A})=wt_{sr}(M)$ which contains $M$.
  As $C\cap\mathcal{A}\neq\{\boldsymbol{0}\}$, due to $(2)$, ${\rm maxsrk}(\mathcal{A})=wt_{sr}(M)\geq \sum\limits_{i=1}^{j-1}n_{i}+\delta+1$.
  Hence, $d_{sr}(C)=\sum\limits_{i=1}^{j-1}n_{i}+\delta+1$ and $C$ is a QMSRD code.
   \end{proof}

In the following, using the isomorphism $Mat_{\mathcal{B}}$ mentioned in Section 2, we derive a proposition which connects the Hamming distance and the sum-rank distance. Let $C\subseteq\mathbb{F}_{q^m}^{n}$ and $A\in\mathbb{F}_{q}^{(n,n)}$, $CA=\{\boldsymbol{c}A\,:\,\boldsymbol{c}\in C\}$. 
If $A$ is invertible, it is obvious that $|C|=|CA|$. Also, we have a known result from \cite{MSK} as below.
\begin{lemma}(\!\cite[Corollary 1.3]{MSK})\label{Lemma 3.6}
  Let $C\subseteq\mathbb{F}_{q^m}^{n}$ and $\mathcal{B}$ be a basis of $\mathbb{F}_{q^m}$ over $\mathbb{F}_{q}$. Then 
  \begin{equation*}
   \begin{aligned}
  d_{sr}(Mat_{\mathcal{B}}(C))= & \min\{d_{H}(CA)\,:\,A=diag(A_{1},\ldots,A_{t})\in\mathbb{F}_{q}^{(n,n)}\\ & \text{and } A_{i}\in\mathbb{F}_{q}^{(n_{i},n_{i})} \text{ is invertible },\,i\in[t]\}.
  \end{aligned}
  \end{equation*}
\end{lemma}
\begin{proposition}\label{Proposition 3.7}
  Let $C\subseteq\mathbb{F}_{q^m}^{n}$ be an $\mathbb{F}_{q}$-linear code and $\mathcal{B}$ be a basis of $\mathbb{F}_{q^m}$ over $\mathbb{F}_{q}$. Then the following are equivalent:
  \begin{itemize}
    \item[(1)] $Mat_{\mathcal{B}}(C)\subseteq\mathbb{F}_{q}^{(m,n_{1}),\ldots,(m,n_{t})}$ is an MSRD code or a QMSRD code.
    \item[(2)] $q^{m(n-d_{H}(CA))}<|CA|\leq q^{m(n-d_{H}(CA)+1)}$ for all $A=diag(A_{1},\ldots,A_{t})\in\mathbb{F}_{q}^{(n,n)}$ such that $A_{i}\in\mathbb{F}_{q}^{(n_{i},n_{i})}$ is invertible, $i\in[t]$.
  \end{itemize}
\end{proposition}
\begin{proof}
  Suppose that $Mat_{\mathcal{B}}(C)$ is an MSRD code or a QMSRD code, then for all $A=diag(A_{1},\ldots,A_{t})\in\mathbb{F}_{q}^{(n,n)}$ such that $A_{i}\in\mathbb{F}_{q}^{(n_{i},n_{i})}$ is invertible, $i\in[t]$, we have
  $$|CA|=|C|=|Mat_{\mathcal{B}}(C)|>q^{m(n-d_{sr}(Mat_{\mathcal{B}}(C)))}.$$
  Due to Lemma \ref{Lemma 3.6}, 
  \begin{equation*}
   \begin{aligned}
  d_{sr}(Mat_{\mathcal{B}}(C))= & \min\{d_{H}(CA)\,:\,A=diag(A_{1},\ldots,A_{t})\in\mathbb{F}_{q}^{(n,n)}\\  & \text{and }A_{i}\in\mathbb{F}_{q}^{(n_{i},n_{i})} \text{ is invertible },\,i\in[t]\}.
  \end{aligned}
  \end{equation*}
  Hence, $$d_{sr}(Mat_{\mathcal{B}}(C))\leq d_{H}(CA).$$
  Combining the two inequalities above with the Singleton bound in the Hamming metric, $$q^{m(n-d_{H}(CA))}<|CA|\leq q^{m(n-d_{H}(CA)+1)}.$$
  
  Conversely, by Lemma \ref{Lemma 3.6}, there exists a matrix $A=diag(A_{1},\ldots,A_{t})\in\mathbb{F}_{q}^{(n,n)}$ where $A_{i}\in\mathbb{F}_{q}^{(n_{i},n_{i})}$ is invertible, $i\in[t]$ such that $d_{sr}(Mat_{\mathcal{B}}(C))=d_{H}(CA)$.
  Since $q^{m(n-d_{H}(CA))}<|CA|\leq q^{m(n-d_{H}(CA)+1)}$, $$q^{m(n-d_{sr}(Mat_{\mathcal{B}}(C)))}<|CA|=|Mat_{\mathcal{B}}(C)|\leq q^{m(n-d_{sr}(Mat_{\mathcal{B}}(C))+1)},$$
  which implies that $Mat_{\mathcal{B}}(C)$ is QMSRD or MSRD.
\end{proof}

In the case where $m_{1}=\cdots=m_{t}=m$ and $n_{1}=\cdots=n_{t}=n$, we can give an upper bound on the block length $t$ of QMSRD codes with Theorem \ref{Theorem 2.7}.
\begin{theorem}\label{Theorem 3.8}
  Suppose there exists a QMSRD code $C\subseteq\mathbb{F}_{q}^{(m,n),\ldots,(m,n)}$ with block length $t$, minimum sum-rank distance $d_{sr}(C)=d\geq 3$ and dimension $\dim(C)=\alpha m+\beta$, where $\alpha$, $\beta$ are integers such that $\alpha=tn-d$ and $1\leq\beta\leq m-1$. Then 
  \begin{equation*}
   \begin{aligned}
  t & \leq \left\lfloor\frac{(q-1)(q^{3m-\beta}-1)+(q^{n}-q^{n\lfloor(d-3)/n\rfloor+n-d+3})(q^{m}-1)}{(q^{n}-1)(q^{m}-1)}\right\rfloor
  +\left\lfloor\frac{d-3}{n}\right\rfloor\\
  & \leq\left\lfloor\frac{(q-1)(q^{3m-\beta}-1)}{(q^{n}-1)(q^{m}-1)}\right\rfloor+1+\left\lfloor\frac{d-3}{n}\right\rfloor.
  \end{aligned}
  \end{equation*}
  Furthermore, 
  \begin{itemize}
    \item[(1)] If $d-3$ is divisible by $n$, then $t\leq\left\lfloor\frac{(q-1)(q^{3m-\beta}-1)}{(q^{n}-1)(q^{m}-1)}\right\rfloor+\frac{d-3}{n}$.
    \item[(2)] If $d\leq n+2$, then $t\leq \left\lfloor\frac{(q-1)(q^{3m-\beta}-1)+(q^{n}-q^{n-d+3})(q^{m}-1)}{(q^{n}-1)(q^{m}-1)}\right\rfloor\leq 1+\left\lfloor\frac{(q-1)(q^{3m-\beta}-1)}{(q^{n}-1)(q^{m}-1)}\right\rfloor$. If in addition $n=m$, then $t\leq 1+\left\lfloor\frac{(q-1)(q^{2n}+q^{n}+1)}{q^{n}-1}\right\rfloor$.
  \end{itemize}
\end{theorem}
\begin{proof}
  Let $d-3=ln+\delta$ with $0\leq \delta<n$. Thus, $l=\left\lfloor\frac{d-3}{n}\right\rfloor$. Then by Theorem \ref{Theorem 2.7}, 
  $$q^{m(tn-d)+\beta}\left[1+\frac{q^{m}-1}{q-1}\left((t-l-1)(q^{n}-1)+q^{n-\delta}-1\right)\right]\leq q^{m(tn-d+3)},$$
  which in turn is equivalent to 
  $$t\leq\frac{(q^{3m-\beta}-1)(q-1)+(q^{n}-q^{n-\delta})(q^{m}-1)}{(q^{m}-1)(q^{n}-1)}+l.$$
  Since $\delta=d-3-\lfloor(d-3)/n\rfloor$ and $n-\delta\geq 1$,
  \begin{equation*}
   \begin{aligned}
  t & \leq \left\lfloor\frac{(q-1)(q^{3m-\beta}-1)+(q^{n}-q^{n\lfloor(d-3)/n\rfloor+n-d+3})(q^{m}-1)}{(q^{n}-1)(q^{m}-1)}\right\rfloor
  +\left\lfloor\frac{d-3}{n}\right\rfloor\\
  & \leq\left\lfloor\frac{(q-1)(q^{3m-\beta}-1)}{(q^{n}-1)(q^{m}-1)}\right\rfloor+1+\left\lfloor\frac{d-3}{n}\right\rfloor.
  \end{aligned}
  \end{equation*}
  \begin{itemize}
    \item[(1)] If $d-3$ is divisible by $n$, then $\delta=0$ and $\left\lfloor\frac{d-3}{n}\right\rfloor=\frac{d-3}{n}$. 
    Hence, $t\leq\left\lfloor\frac{(q-1)(q^{3m-\beta}-1)}{(q^{n}-1)(q^{m}-1)}\right\rfloor+\frac{d-3}{n}$.
    \item[(2)] If $d\leq n+2$, then $\left\lfloor\frac{d-3}{n}\right\rfloor=0$. 
    Therefore, $t\leq \left\lfloor\frac{(q-1)(q^{3m-\beta}-1)+(q^{n}-q^{n-d+3})(q^{m}-1)}{(q^{n}-1)(q^{m}-1)}\right\rfloor$\\$\leq 1+\left\lfloor\frac{(q-1)(q^{3m-\beta}-1)}{(q^{n}-1)(q^{m}-1)}\right\rfloor$. 
    If we also let $n=m$, then $t\leq 1+\left\lfloor\frac{(q-1)(q^{3n-\beta}-1)}{(q^{n}-1)^2}\right\rfloor
    \leq 1+\left\lfloor\frac{(q-1)(q^{3n}-1)}{(q^{n}-1)^2}\right\rfloor=1+\left\lfloor\frac{(q-1)(q^{2n}+q^{n}+1)}{q^{n}-1}\right\rfloor$.
  \end{itemize}
\end{proof}
\begin{remark}
  {\rm Note that this theorem can be regarded as a generalization of \cite[Theorem 7]{MR}, which provided an upper bound of the folded length of a linear code $C$ in the folded Hamming distance. Let $n=1$, we can derive \cite[Theorem 7]{MR} from this theorem.}
\end{remark}

\begin{example}
In \cite[Theorem 11]{LCJXZ}, the authors constructed a linear sum-rank-metric code $C\subseteq\mathbb{F}_{3}^{(2,2),\ldots,(2,2)}$ with block length $3^2$, minimum sum-rank distance $d_{sr}(C)=4$ and dimension $\dim(C)=2(2\cdot 3^2 -4)=2(2\cdot 3^2 -d_{sr}(C))=28$. However, there exists no QMSRD code $C\subseteq\mathbb{F}_{q}^{(2,2),\ldots,(2,2)}$ with block length $3^2$, minimum sum-rank distance $d_{sr}(C)=4$ and dimension $\dim(C)=2(2\cdot 3^2 -d_{sr}(C)+1)-1=29$. Suppose that there exists a linear sum-rank-metric code $C\subseteq\mathbb{F}_{3}^{(2,2),\ldots,(2,2)}$ with block length $t$, minimum sum-rank distance $d_{sr}(C)=4$ and dimension $\dim(C)=2(2t-d_{sr}(C)+1)-1$.  By Theorem \ref{Theorem 3.8}, $t\leq 1+\left\lfloor\frac{(3-1)(3^{3\cdot 2-1}-1)}{(3^{2}-1)^2}\right\rfloor=8$.  
\end{example}

Similar to Hamming-metric codes, we consider puncturing sum-rank-metric codes. By removing some columns of a known QMSRD codes, we obtain some new QMSRD codes.
\begin{theorem}\label{Theorem 3.11}
  Let $C\subseteq\mathbb{F}_{q}^{(m_{1},n_{1}),\ldots,(m_{t},n_{t})}$ be a QMSRD code. Its minimum sum-rank distance and dimension are given as follows: 
  $$d_{sr}(C)=\sum\limits_{i=1}^{j-1} n_{i}+\delta+1\geq 2, \quad \dim(C)=\sum\limits_{i=j}^{t} m_{i}n_{i}-\delta m_{j}-s,$$
  where $j\in[t]$, $0\leq \delta\leq n_{j}-1$ and $0< s\leq m_{j}-1$ are integers. If $\delta\neq 0$, choose $r\in[j]$ and if $\delta=0$, choose $r\in[j-1]$. 
  Set $$\tilde{n_{i}}=\begin{cases}
          n_{i}, & \mbox{if } i\neq r \\
          n_{r}-1, & \mbox{if } i=r.
        \end{cases}$$
Then we have a QMSRD code $C^{\prime}\subseteq\mathbb{F}_{q}^{(m_{1},\tilde{n_{1}}),\ldots,(m_{t},\tilde{n_{t}})}$, which has minimum sum-rank distance and dimension
$$d_{sr}(C^{\prime})=d_{sr}(C)-1, \quad \dim(C^{\prime})=\dim(C).$$
\end{theorem}
\begin{proof}
  Set the map 
  $$\pi\,:\, C\rightarrow\mathbb{F}_{q}^{(m_{1},\tilde{n_{1}}),\ldots,(m_{t},\tilde{n_{t}})},\quad (M_{1},\ldots,M_{t})\mapsto(M_{1},\ldots,\widetilde{M_{r}},\ldots,M_{t}),$$
  where $\widetilde{M_{r}}\in\mathbb{F}_{q}^{(m_{r},\tilde{n_{r}})}=\mathbb{F}_{q}^{(m_{r},n_{r}-1)}$ is obtained by deleing the last column of $M_{r}$. 
  Since $d_{sr}(C)\geq 2$, the map is injective. Let $C^{\prime}=\pi(C)$. If $r\in[j-1]$, 
  \begin{equation*}
    \begin{aligned}
    \dim(C^{\prime}) & =\dim(C)=\sum\limits_{i=j}^{t}m_{i}n_{i}-\delta m_{j}-s=\sum\limits_{i=j}^{t}m_{i}\tilde{n_{i}}-\delta m_{j}-s,\\
    d_{sr}(C^{\prime}) & =d_{sr}(C)-1=\sum\limits_{i=1}^{j-1} n_{i}+\delta=\sum\limits_{i=1}^{j-1} \tilde{n_{i}}+\delta+1.
    \end{aligned}
  \end{equation*}
Thus, $C^{\prime}$ is a QMSRD code. If $r=j$, then $\delta>0$. At this case, 
\begin{equation*}
    \begin{aligned}
    \dim(C^{\prime}) & =\dim(C)=\sum\limits_{i=j}^{t}m_{i}n_{i}-\delta m_{j}-s=\sum\limits_{i=j}^{t}m_{i}\tilde{n_{i}}-(\delta-1)m_{j}-s,\\
    d_{sr}(C^{\prime}) & =d_{sr}(C)-1=\sum\limits_{i=1}^{j-1} n_{i}+\delta=\sum\limits_{i=1}^{j-1} \tilde{n_{i}}+(\delta-1)+1.
    \end{aligned}
  \end{equation*}
Thus, $C^{\prime}$ is a QMSRD code as well.
\end{proof}

\section{Dually Quasi-MSRD Codes}
In this section, we consider the dual code of a QMSRD code. It is known that when $m_{1}=\cdots=m_{t}$, $C\subseteq\mathbb{F}_{q}^{(m_{1},n_{1}),\ldots,(m_{t},n_{t})}$ is an MSRD code if and only if its dual code $C^{\bot}$ is an MSRD code, 
but for arbitrary $m_{i}$s, the result is false in general, see \cite{BGR}. 
Unfortunately, even when $m_{1}=\cdots=m_{t}$, the dual code of a QMSRD code $C\subseteq\mathbb{F}_{q}^{(m_{1},n_{1}),\ldots,(m_{t},n_{t})}$ may be not a QMSRD code. The following is an example.
\begin{example}
  Let $m_{t}>1$ and $C\subseteq\mathbb{F}_{q}^{(m_{1},n_{1}),\ldots,(m_{t},n_{t})}$ be a linear sum-rank-metric code with dimension $1\leq\dim(C)<m_{t}$ and minimum sum-rank distance $d<n=n_{1}+\cdots+n_{t}$. Since $\dim(C)=m_{t}n_{t}-m_{t}(n_{t}-1)-s$ and $d_{sr}(C)<\sum\limits_{i=1}^{t-1}n_{i}+(n_{t}-1)+1=n$, $1\leq s\leq m_{t}-1$, $C$ is not a QMSRD code. 
  Since $1\leq\dim(C)<m_{t}\leq m_{1}$, $\dim(C^{\bot})=\sum\limits_{i=1}^{t}m_{i}n_{i}-s^{\prime}$, where $1\leq s^{\prime}\leq m_{1}-1$. 
   By the Singleton bound, $d_{sr}(C^{\bot})=1$. Hence, $C^{\bot}$ is a QMSRD code. 
\end{example}
Based on the fact above, we give the definition below.
\begin{definition}
  {\rm Let $C\subseteq\mathbb{F}_{q}^{(m_{1},n_{1}),\ldots,(m_{t},n_{t})}$ be a linear sum-rank-metric code. If both $C$ and $C^{\bot}$ are QMSRD, $C$ is a \emph{dually quasi-MSRD (dually QMSRD)} code.}
\end{definition}
In Example \ref{Example 3.4}, we find a dually QMSRD code. 

\subsection{\textbf{Characterization of Dually QMSRD Codes}}
In this subsection, we try to characterize dually QMSRD codes. Firstly, we recall two results about MSRD codes from \cite{CGLGMS}.
\begin{lemma}\cite[Proposition VII.11]{CGLGMS}\label{Lemma 4.3}
  Let $C\subseteq\mathbb{F}_{q}^{(m_{1},n_{1}),\ldots,(m_{t},n_{t})}$ be a non-trivial linear sum-rank-metric code and $n=n_{1}+\cdots+n_{t}$. 
  Then $d_{sr}(C)+d_{sr}(C^{\bot})=n+2$ if and only if both $C$ and $C^{\bot}$ are MSRD.
\end{lemma} 
\begin{lemma}\cite[Proposition VII.13]{CGLGMS}\label{Lemma 4.4}
If there exists a non-trivial linear sum-rank-metric code $C\subseteq\mathbb{F}_{q}^{(m_{1},n_{1}),\ldots,(m_{t},n_{t})}$ such that $n=n_{1}+\cdots+n_{t}$ and $d_{sr}(C)+d_{sr}(C^{\bot})=n+2$, then $m_{1}=\cdots=m_{t}$.
\end{lemma}
Combining these results with Theorem \ref{Theorem 2.5} and Proposition \ref{Proposition 2.6}, we get the proposition below.
\begin{proposition}\label{Proposition 4.5}
  Let $C\subseteq\mathbb{F}_{q}^{(m_{1},n_{1}),\ldots,(m_{t},n_{t})}$ be a non-trivial linear sum-rank-metric code of dimension $\dim(C)=\sum\limits_{i=j}^{t}m_{i}n_{i}-\delta m_{j}-s$, where $j\in[t]$, $0\leq\delta\leq n_{j}-1$ and $0\leq s\leq m_{j}-1$. Let $n=n_{1}+\cdots+n_{t}$. Then the following holds:
  \begin{itemize}
    \item[(1)] If $m_{1}=\cdots=m_{t}$ and $s=0$, then $d_{sr}(C)+d_{sr}(C^{\bot})\leq n$ or $d_{sr}(C)+d_{sr}(C^{\bot})=n+2$.
    \item[(2)] If $m_{i}\neq m_{j}$ for some $i$, $j$ with $i\neq j$ or $s\neq 0$, then $d_{sr}(C)+d_{sr}(C^{\bot})\leq n+1$.
  \end{itemize} 
\end{proposition}
\begin{proof}
  \begin{itemize}
    \item[(1)] By Theorem \ref{Theorem 2.5}, we have $$d_{sr}(C)\leq\sum\limits_{i=1}^{j-1}n_{i}+\delta+1.\qquad (\ast)$$
    And according to Proposition \ref{Proposition 2.6}, as $\dim(C^{\bot})=\sum\limits_{i=1}^{j-1}m_{i}n_{i}+\delta m_{j}+s$, 
    $$d_{sr}(C^{\bot})\leq\sum\limits_{i=j}^{t}n_{i}-\delta+\begin{cases}
                                                              1, & \mbox{if } s=0, \\
                                                              0, & \mbox{else}.
                                                            \end{cases}\qquad (\star)$$
    Hence, $d_{sr}(C)+d_{sr}(C^{\bot})\leq n+2$.
    Since $m_{1}=\cdots=m_{t}$, $C$ is an MSRD code if and only if $C^{\bot}$ is an MSRD code.
    At that case, by Lemma \ref{Lemma 4.3}, $d_{sr}(C)+d_{sr}(C^{\bot})=n+2$. 
    Otherwise, both $C$ and $C^{\bot}$ are not MSRD codes. Since $s=0$ and $m_{1}=\cdots=m_{t}$, $\dim(C)\leq m(n-d_{sr}(C))$ and $\dim(C^{\bot})\leq m(n-d_{sr}(C^{\bot}))$, which implies that $d_{sr}(C)+d_{sr}(C^{\bot})\leq n$.
    \item[(2)] If $m_{i}\neq m_{j}$ for some $i$, $j$ with $i\neq j$, by Lemma \ref{Lemma 4.4}, it is impossible that $d_{sr}(C)+d_{sr}(C^{\bot})=n+2$. Hence, we have $d_{sr}(C)+d_{sr}(C^{\bot})\leq n+1$. If $s\neq 0$, then by inequalities $(\ast)$ and $(\star)$, $d_{sr}(C)+d_{sr}(C^{\bot})\leq n+1$. 
  \end{itemize}
\end{proof}
Furthermore, we have the following proposition, which plays a vital role later.
\begin{proposition}\label{Proposition 4.6}
  Let $C\subseteq\mathbb{F}_{q}^{(m,n_{1}),\ldots,(m,n_{t})}$ be a linear sum-rank-metric code. Then $C$ is a dually QMSRD code if and only if $d_{sr}(C)+d_{sr}(C^{\bot})=n+1$.
\end{proposition}
\begin{proof}
  Let $C$ be a dually QMSRD code. Then $$d_{sr}(C)=n-\left\lceil\frac{\dim(C)}{m}\right\rceil+1$$ and $$d_{sr}(C^{\bot})=n-\left\lceil\frac{mn-\dim(C)}{m}\right\rceil+1=\left\lfloor\frac{\dim(C)}{m}\right\rfloor+1,$$
  where $n=n_{1}+\cdots+n_{t}$. Since $\dim(C)$ is not divisible by $m$, $d_{sr}(C)+d_{sr}(C^{\bot})=n+1$.
  
  Suppose $d_{sr}(C)+d_{sr}(C^{\bot})=n+1$, then by Proposition \ref{Proposition 4.5}, $m$ is not divisible by $\dim(C)$. In addition, by Theorem \ref{Theorem 2.5}, $$d_{sr}(C)\leq n-\left\lceil\frac{\dim(C)}{m}\right\rceil+1  \text{ and } d_{sr}(C^{\bot})\leq \left\lfloor\frac{\dim(C)}{m}\right\rfloor+1.$$
  Hence, both of the inequalities must be equalities and $C$ is a dually QMSRD code.
\end{proof}
\begin{remark}
  {\rm For arbitrary $m_{i}$s, Proposition \ref{Proposition 4.6} is generally incorrect. Precisely, for $C\subseteq\mathbb{F}_{q}^{(m_{1},n_{1}),\ldots,(m_{t},n_{t})}$, if $d_{sr}(C)+d_{sr}(C^{\bot})=n+1$, $C$ may be not a dually QMSRD code. 
  In \cite[Example VI.2]{BGR}, the authors provided a counter example. 
  We also find that $d_{sr}(C)+d_{sr}(C^{\bot})=n+1$ may be false even if $C$ is a dually QMSRD code. Below we present an example.}
\end{remark}
\begin{example}
  Let $C\subseteq\mathbb{F}_{2}^{(3,2),(1,1)}$ be a linear sum-rank-metric code generated by the following two elements: 
  $$\left(\begin{pmatrix}
            1 & 0 \\
            0 & 1 \\
            0 & 0 
          \end{pmatrix},\,(0)\right) \text{ and }
  \left(\begin{pmatrix}
            0 & 0 \\
            1 & 0 \\
            0 & 1 
          \end{pmatrix},\,(0)\right).$$
Then $\dim(C)=2$ and $d_{sr}(C)=2$, which shows that $C$ is a QMSRD code.
What is more, $\dim(C^{\bot})=5$ and $d_{sr}(C^{\bot})=1$. Hence, $C^{\bot}$ is a QMSRD code as well.
In a word, $C$ is a dually QMSRD code, but $d_{sr}(C)+d_{sr}(C^{\bot})=3\neq 4=3+1$.
\end{example}
To investigate dually QMSRD codes more deeply, it is necessary to revisit some other knowledge in \cite{BGR}. 
\begin{definition}
  {\rm For $i\in[t]$, let $\mathcal{L}_{i}$ be the lattice of subspaces of $\mathbb{F}_{q}^{n_{i}}$ partially ordered by inclusion. Let $\mathcal{L}=\mathcal{L}_{1}\times\cdots\times\mathcal{L}_{t}$ be the product lattice endowed with the product order, which we denote by $\subseteq$.
  For any $U=(U_{1},\ldots,U_{t})\in\mathcal{L}$, the rank function and dimension of $U$ is defined as ${\rm rk}_{\mathcal{L}}(U)=\sum\limits_{i=1}^{t}\dim(U_{i})$ and $\dim(U)=(\dim(U_{1}),\ldots,\dim(U_{t}))$ respectively. 
  Finally, for $U\subseteq V$, the M\"{o}bius function of $\mathcal{L}$ is given by 
  $$\mu_{\mathcal{L}}(U,V)=\prod\limits_{i=1}^{t}(-1)^{\dim(V_{i})-\dim(U_{i})}q^{\binom{\dim(V_{i})-\dim(U_{i})}{2}}.$$}
\end{definition}
\begin{definition}
{\rm For a matrix $M\in\mathbb{F}_{q}^{(m,n)}$, we let ${\rm rowsp}(M)\subseteq\mathbb{F}_{q}^{n}$ denote its row space. We define the \emph{support} as: 
$$\sigma\,:\, \mathbb{F}_{q}^{(m_{1},n_{1}),\ldots,(m_{t},n_{t})}\rightarrow\mathcal{L},\,M=(M_{1},\ldots,M_{t})\mapsto({\rm rowsp}(M_{1}),\ldots,{\rm rowsp}(M_{t})).$$}
\end{definition}
Note that $wt_{sr}(M)={\rm rk}_{\mathcal{L}}(\sigma(M))$ for all $M\in\mathbb{F}_{q}^{(m_{1},n_{1}),\ldots,(m_{t},n_{t})}$. 
Furthermore, for a linear sum-rank code $C\subseteq\mathbb{F}_{q}^{(m_{1},n_{1}),\ldots,(m_{t},n_{t})}$ and $U\in\mathcal{L}$, we define 
$$C(U)=\{M\in C\,:\,\sigma(M)\subseteq U\}\subseteq\mathbb{F}_{q}^{(m_{1},n_{1}),\ldots,(m_{t},n_{t})}.$$
Based on this definition, the authors in \cite{BGR} gave the following results. 
\begin{lemma}(\!\cite[Proposition II.7]{BGR})\label{Lemma 4.11}
Let $\{\boldsymbol{0}\}\neq C\in\mathbb{F}_{q}^{(m_{1},n_{1}),\ldots,(m_{t},n_{t})}$ be a linear sum-rank-metric code and $1\leq d\leq n=n_{1}+\cdots+n_{t}$ be an integer. Then 
$d_{sr}(C)\geq d$ if and only if $|C(U)|=1$ for all $U\in\mathcal{L}$ with $rk_{\mathcal{L}}(U)=d-1$, which is equivalent to that $|C(U)|=1$ for all $U\in\mathcal{L}$ with $rk_{\mathcal{L}}(U)\leq d-1$.
\end{lemma}
\begin{lemma}(\!\cite[Proposition V.3]{BGR})\label{Lemma 4.12}
Let $C\subseteq\mathbb{F}_{q}^{(m_{1},n_{1}),\ldots,(m_{t},n_{t})}$ be a linear sum-rank-metric code and $U=(U_{1},\ldots,U_{t})\in\mathcal{L}$ with $\dim(U_{i})=u_{i}$ for $i\in[t]$. Then
$$|C(U)|=\frac{|C|}{q^{\sum\limits_{i=1}^{t}m_{i}(n_{i}-u_{i})}}|C^{\bot}(U^{\bot})|,$$
where $U^{\bot}=(U_{1}^{\bot},\ldots,U_{t}^{\bot})$.
\end{lemma}
What is more, the authors in \cite{BGR} gave the following definition.
\begin{definition}\label{Definition 4.13}
  {\rm Let $C\subseteq\mathbb{F}_{q}^{(m_{1},n_{1}),\ldots,(m_{t},n_{t})}$ be a linear sum-rank-metric code. For $r\in\mathbb{N}_{0}$, $\boldsymbol{u}=(u_{1},\ldots,u_{t})\in\mathbb{N}_{0}^{t}$ and $U\in\mathcal{L}$, we let
   \begin{equation*}
   \begin{aligned}
   W_{r}(C) &=|\{M\in C\,:\,wt_{sr}(C)=r\}|,\\
   W_{\boldsymbol{u}}(C) &=|\{M\in C\,:\,{\rm rank}(M_{i})=u_{i} \text{ for all } i\in[t]\}|,\\
   W_{U}(C) &=|\{M\in C\,:\,\sigma(M)=U\}|.
   \end{aligned}
   \end{equation*}
   Then the  lists 
   $$(W_{r}(C))_{r\in\mathbb{N}_{0}},\quad (W_{\boldsymbol{u}}(C))_{\boldsymbol{u}\in\mathbb{N}_{0}^{t}},\quad (W_{U}(C))_{U\in\mathcal{L}}$$
   are called the \emph{sum-rank, rank-list and support distributions} of $C$ respectively.}
\end{definition}
In the following, we characterize dually QMSRD codes in terms of their number of minimum sum-rank weight codewords with Proposition \ref{Proposition 4.6} and Lemmas \ref{Lemma 4.11}, \ref{Lemma 4.12}.
\begin{theorem}\label{Theorem 4.14}
  Let $C\subseteq\mathbb{F}_{q}^{(m,n_{1}),\ldots,(m,n_{t})}$ be a QMSRD code of minimum sum-rank distance $d_{sr}(C)=d$ and dimension $\dim(C)=\alpha m+\beta$, where $n=n_{1}+\cdots+n_{t}$, $1\leq \beta\leq m-1$. $C$ is dually QMSRD if and only if 
  $$W_{d}(C)=(q^{\beta}-1)\sum\limits_{\substack{(u_{1},\ldots,u_{t})\in\mathbb{N}_{0}^{t}\\u_{1}+\cdots+u_{t}=d}}\prod\limits_{i=1}^{t}\binom{n_{i}}{u_{i}}_{q}.$$
\end{theorem}
\begin{proof}
 Since $C$ is a QMSRD code, $d=n-\left\lceil\frac{\dim(C)}{m}\right\rceil+1=n-\left\lceil\frac{\alpha m+\beta}{m}\right\rceil+1=n-\alpha$, which implies that $\alpha =n-d$. 
  Let $U=(U_{1},\ldots,U_{t})\in\mathcal{L}$ and ${\rm rk}_{\mathcal{L}}(U)=d$.
  By Lemma \ref{Lemma 4.12}, 
  \begin{equation*}
   \begin{aligned}
  |C(U)|& =1+W_{U}(C)=\frac{|C|}{q^{m(n-d)}}|C^{\bot}(U^{\bot})|\\ &=\frac{q^{m(n-d)+\beta}}{q^{m(n-d)}}|C^{\bot}(U^{\bot})|=q^{\beta}|C^{\bot}(U^{\bot})|.
  \quad(\dagger)
  \end{aligned}
  \end{equation*}
  Since $C$ is a QMSRD code, $C$ is a dually QMSRD code if and only if $C^{\bot}$ is a QMSRD code, which by Proposition \ref{Proposition 4.6}, implies that $d_{sr}(C^{\bot})=n-d+1$. \\
  Furthermore, by Lemma \ref{Lemma 4.11}, $d_{sr}(C^{\bot})=n-d+1$ if and only if $|C^{\bot}(U^{\bot})|=1$ for all $U\in\mathcal{L}$ with ${\rm rk}_{\mathcal{L}}(U)=d$, which by the equality $(\dagger)$, is equivalent to the fact that $W_{U}(C)=q^{\beta}-1$ for all $U\in\mathcal{L}$ with ${\rm rk}_{\mathcal{L}}(U)=d$. 
  $W_{U}(C)=q^{\beta}-1$ for all $U\in\mathcal{L}$ with ${\rm rk}_{\mathcal{L}}(U)=d$ if and only if $$W_{d}(C)=\sum\limits_{U\in\mathcal{L},\,{\rm rk}_{\mathcal{L}}(U)=d}(q^{\beta}-1)=(q^{\beta}-1)\sum\limits_{\substack{(u_{1},\ldots,u_{t})\in\mathbb{N}_{0}^{t}\\u_{1}+\cdots+u_{t}=d}}\prod\limits_{i=1}^{t}\binom{n_{i}}{u_{i}}_{q}.$$
  Hence, we get the conclusion that $C$ is a dually QMSRD code if and only if 
  $$W_{d}(C)=(q^{\beta}-1)\sum\limits_{\substack{(u_{1},\ldots,u_{t})\in\mathbb{N}_{0}^{t}\\u_{1}+\cdots+u_{t}=d}}\prod\limits_{i=1}^{t}\binom{n_{i}}{u_{i}}_{q}.$$
 \end{proof}
 \begin{remark}
  {\rm In \cite[Proposition 16]{CGLR}, the authors provided a characterization of dually QMRD codes in terms of the number of codewords with minimum weight. 
  Theorem \ref{Theorem 4.14} can be regarded as a generalization of the proposition.}
 \end{remark}
 \begin{example}
   For the dually QMSRD code in Example \ref{Example 3.4}, $W_{3}(C)=6=(2-1)\cdot 2\binom{2}{1}_{2}\binom{2}{2}_{2}$, which agrees with the theorem above.
 \end{example}
 
 \subsection{Weight Distribution of Dually QMSRD Codes}
 In this subsection, we study the weight distribution of dually QMSRD codes including support, rank-list and sum-rank distributions in Definition \ref{Definition 4.13}. We begin with a lemma.
 \begin{lemma}\label{Lemma 4.16}
   Let $C\subseteq\mathbb{F}_{q}^{(m,n_{1}),\ldots,(m,n_{t})}$ be a non-zero dually QMSRD code with minimum sum-rank distance $d_{sr}(C)=d$ and dimension $\dim(C)=\alpha m+\beta$, where $n=n_{1}+\cdots+n_{t}$, $\alpha=n-d$ and $1\leq \beta\leq m-1$. For all $U\in\mathcal{L}$ with ${\rm rk}_{\mathcal{L}}(U)=u$, we have
   $$|C(U)|=\begin{cases}
              1, & \mbox{if } u<d, \\
              q^{m(u-d)+\beta}, & \mbox{else}.
            \end{cases}$$
 \end{lemma}
 \begin{proof}
   For $u<d$, we can get the conclusion directly from Lemma \ref{Lemma 4.11}. Thus, suppose that $u\geq d$. Then ${\rm rk}_{\mathcal{L}}(U^{\bot})=n-u\leq n-d$. By Proposition \ref{Proposition 4.6}, $d_{sr}(C^{\bot})=n+1-d$, so by Lemma \ref{Lemma 4.11}, $|C^{\bot}(U^{\bot})|=1$. Finally, since $\alpha=n-d$, according to Lemma \ref{Lemma 4.12}, we have
   $$|C(U)|=\frac{|C|}{q^{m(n-u)}}|C^{\bot}(U^{\bot})|=\frac{q^{m(n-d)+\beta}}{q^{m(n-u)}}=q^{m(u-d)+\beta}.$$ 
 \end{proof}
 \vspace{-2mm}
 Then we can derive the support distribution of dually QMSRD codes. Let $\boldsymbol{u}=(u_{1},\ldots,u_{t})$, $\boldsymbol{v}=(v_{1},\ldots,v_{t})\in\mathbb{N}_{0}^{t}$. If $v_{i}\leq u_{i}$ for all $i\in[t]$, we say $\boldsymbol{v}\leq\boldsymbol{u}$.
 \begin{theorem}\label{Theorem 4.18}
   Let $C\subseteq\mathbb{F}_{q}^{(m,n_{1}),\ldots,(m,n_{t})}$ be a non-zero dually QMSRD code with minimum sum-rank distance $d_{sr}(C)=d$ and dimension $\dim(C)=\alpha m+\beta$, where $n=n_{1}+\cdots+n_{t}$, $\alpha=n-d$ and $1\leq \beta\leq m-1$. Then for any $U\in\mathcal{L}\setminus 0$ with $\dim(U)=\boldsymbol{u}=(u_{1},\ldots,u_{t})$ and ${\rm rk}_{\mathcal{L}}(U)=u$, 
   $$W_{U}(C)=\sum\limits_{l=d}^{u}(q^{m(l-d)+\beta}-1)f_{l}(\boldsymbol{u}),$$\\
   where $f_{l}(\boldsymbol{u})=\sum\limits_{\substack{\boldsymbol{v}\leq\boldsymbol{u}\\v_{1}+\cdots+v_{t}=l}}\prod\limits_{i=1}^{t}(-1)^{u_{i}-v_{i}}q^{\binom{u_{i}-v_{i}}{2}}\binom{u_{i}}{v_{i}}_{q}$. 
   In particular, $W_{U}(C)$ only depends on $\boldsymbol{u}=\dim(U)$, but not on the subspace $U$ itself.
 \end{theorem} 
 \begin{proof}
   Set $\mathcal{V}_{l}(U)=\{V\in\mathcal{L}\,:\,V\subseteq U,\, {\rm rk}_{\mathcal{L}}(V)=l\}$. 
   Since $C(U)=\sum\limits_{V\subseteq U}W_{V}(C)$ and $\sum\limits_{V\subseteq U}\mu_{\mathcal{L}}(V,U)=0$, for any $U\in\mathcal{L}\setminus 0$, applying the M\"{o}bius inversion along with Lemma \ref{Lemma 4.16}, 
   \begin{equation*}
   \begin{aligned}
   W_{U}(C) & =\sum\limits_{\substack{V\in\mathcal{L}\\V\subseteq U}}|C(V)|\mu_{\mathcal{L}}(V,U)\\
   & = \sum\limits_{l=0}^{d-1}\sum\limits_{V\in\mathcal{V}_{l}(U)}\mu_{\mathcal{L}}(V,U)
   +\sum\limits_{l=d}^{u}\sum\limits_{V\in\mathcal{V}_{l}(U)}q^{m(l-d)+\beta}\mu_{\mathcal{L}}(V,U)\\
   & = \sum\limits_{l=d}^{u}(q^{m(l-d)+\beta}-1)\sum\limits_{V\in\mathcal{V}_{l}(U)}\mu_{\mathcal{L}}(V,U)\\
   & = \sum\limits_{l=d}^{u}(q^{m(l-d)+\beta}-1)\sum\limits_{\substack{\boldsymbol{v}\leq\boldsymbol{u}\\v_{1}+\cdots+v_{t}=l}}\prod\limits_{i=1}^{t}(-1)^{u_{i}-v_{i}}q^{\binom{u_{i}-v_{i}}{2}}\binom{u_{i}}{v_{i}}_{q}\\
   & = \sum\limits_{l=d}^{u}(q^{m(l-d)+\beta}-1)f_{l}(\boldsymbol{u}).
   \end{aligned}
   \end{equation*} 
 \end{proof}
 Next, we can obtain the rank-list distribution and sum-rank distribution of dually QMSRD codes immediately. Consider the space $\mathbb{F}_{q}^{(m,n_{1}),\ldots,(m,n_{t})}$, we define the index set
 $$\mathcal{I}_{(n_{1},\ldots,n_{t})}=\{(u_{1},\ldots,u_{t})\in\mathbb{N}_{0}^{t}\,:\,u_{i}\leq n_{i} \text{ for all } i\in[t]\}.$$
 Summing all $W_{U}(C)$s with $\dim(U)=\boldsymbol{u}\in\mathcal{I}_{(n_{1},\ldots,n_{t})}\setminus\{\boldsymbol{0}\}$, we can compute the value of $W_{\boldsymbol{u}}(C)$.
 \begin{corollary}\label{Corollary 4.19}
    Let $C\subseteq\mathbb{F}_{q}^{(m,n_{1}),\ldots,(m,n_{t})}$ be a non-zero dually QMSRD code with minimum sum-rank distance $d_{sr}(C)=d$ and dimension $\dim(C)=\alpha m+\beta$, where $n=n_{1}+\cdots+n_{t}$, $\alpha=n-d$ and $1\leq \beta\leq m-1$. Then for any $\boldsymbol{u}=(u_{1},\ldots,u_{t})\in\mathcal{I}_{(n_{1},\ldots,n_{t})}\setminus\{\boldsymbol{0}\}$ with $\sum\limits_{i=1}^{t}u_{i}=u$, we have
    $$W_{\boldsymbol{u}}(C)=\sum\limits_{l=d}^{u}(q^{m(l-d)+\beta}-1)g_{l}(\boldsymbol{u}),$$
    where $g_{l}(\boldsymbol{u})=\sum\limits_{\substack{\boldsymbol{v}\leq\boldsymbol{u}\\v_{1}+\cdots+v_{t}=l}}\prod\limits_{i=1}^{t}(-1)^{u_{i}-v_{i}}q^{\binom{u_{i}-v_{i}}{2}}\binom{u_{i}}{v_{i}}_{q}\binom{n_{i}}{u_{i}}_{q}$.
 \end{corollary}
 We can get the value of $W_{r}(C)$ by summing all $W_{U}(C)$s with ${\rm rk}_{\mathcal{L}}(U)=r$.
 \begin{corollary}\label{Corollary 4.20}
   Let $C\subseteq\mathbb{F}_{q}^{(m,n_{1}),\ldots,(m,n_{t})}$ be a non-zero dually QMSRD code with minimum sum-rank distance $d_{sr}(C)=d$ and dimension $\dim(C)=\alpha m+\beta$, where $n=n_{1}+\cdots+n_{t}$, $\alpha=n-d$ and $1\leq \beta\leq m-1$. Then 
   $$W_{r}(C)=\sum\limits_{l=d}^{r}(q^{m(l-d)+\beta}-1)\sum\limits_{\substack{\boldsymbol{u}\in\mathcal{I}_{(n_{1},\ldots,n_{t})}\\u_{1}+\cdots+u_{t}=r}}g_{l}(\boldsymbol{u}),$$
   where $g_{l}(\boldsymbol{u})=\sum\limits_{\substack{\boldsymbol{v}\leq\boldsymbol{u}\\v_{1}+\cdots+v_{t}=l}}\prod\limits_{i=1}^{t}(-1)^{u_{i}-v_{i}}q^{\binom{u_{i}-v_{i}}{2}}\binom{u_{i}}{v_{i}}_{q}\binom{n_{i}}{u_{i}}_{q}$.
 \end{corollary}
 \begin{example}
   Using the corollary above, we can compute the sum-rank distribution of the dually QMSRD code in Example \ref{Example 3.4} as follows.
   \begin{equation*}
   \begin{aligned}
   & W_{0}(C)=1,\, W_{3}(C)=(2-1)\sum\limits_{\substack{\boldsymbol{u}\in\mathcal{I}_{(2,2)}\\u_{1}+u_{2}=3}}\binom{2}{u_{1}}_{2}\binom{2}{u_{2}}_{2}=6,\\ 
   & W_{4}(C)=(2-1)\left[\sum\limits_{\substack{\boldsymbol{v}\leq(2,2)\\v_{1}+v_{2}=3}}(-1)\binom{2}{v_{1}}_{2}\binom{2}{v_{2}}_{2}\right]+(2^3 -1)=1.
   \end{aligned}
   \end{equation*}
   After verification, the results agree with the fact.
 \end{example}
 Now, we arrive at a non-existence criterion for dually QMSRD codes, which can be seen as an application of the weight distribution of dually QMSRD codes. 
 \begin{proposition}\label{Proposition 20}
   Suppose there exists a dually QMSRD code $C\subseteq\mathbb{F}_{q}^{(m,n_{1}),\ldots,(m,n_{t})}$ with minimum sum-rank distance $d_{sr}(C)=d$ and dimension $\dim(C)=\alpha m+\beta$, where $n=n_{1}+\cdots+n_{t}$, $\alpha=n-d$ and $1\leq \beta\leq m-1$. Then for all $\boldsymbol{u}=(u_{1},\ldots,u_{t})\in\mathcal{I}_{(n_{1},\ldots,n_{t})}\setminus\{\boldsymbol{0}\}$ with $\sum\limits_{i=1}^{t}u_{i}=u$,
   $$\omega(\boldsymbol{u})=\sum\limits_{l=d}^{u}(q^{m(l-d)+\beta}-1)f_{l}(\boldsymbol{u})\geq 0,$$ 
   where $f_{l}(\boldsymbol{u})=\sum\limits_{\substack{\boldsymbol{v}\leq\boldsymbol{u}\\v_{1}+\cdots+v_{t}=l}}
   \prod\limits_{i=1}^{t}(-1)^{u_{i}-v_{i}}q^{\binom{u_{i}-v_{i}}{2}}\binom{u_{i}}{v_{i}}_{q}$.
 \end{proposition}
Due to the existence of $C$, $W_{U}(C)\geq 0$ for any $U\in\mathcal{L}\setminus 0$. Then by Theorem \ref{Theorem 4.18}, it is easy to get the result. 
Using this proposition, we can prove the non-existence of dually QMSRD codes for some parameters. The following is an example.
\begin{example}
  We find that there exists no dually QMSRD code with minimum sum-rank distance $7$ and dimension $5$ in $\mathbb{F}_{2}^{(3,3),(3,3),(3,2)}$. 
  Suppose that such code exists, we let $\boldsymbol{u}=(3,3,2)$. However, $\omega(\boldsymbol{u})=(2^{3+2}-1)+(2^{2}-1)[(-1)\binom{3}{2}_{2}\cdot 2+(-1)\binom{2}{1}_{2}]=-20$, violating the Proposition \ref{Proposition 20}.
\end{example}

\subsection{Generalized Weights of Dually QMSRD Codes}
In this subsection, we consider the generalized weights of dually QMSRD codes, mentioned in Section 2. Firstly, we review some basic properties of generalized sum-rank weights.
\begin{lemma}(\!\cite[Proposition VI.6]{CGLGMS})\label{Lemma 4.21}
Let $C\subseteq \mathbb{F}_{q}^{(m_{1},n_{1}),\ldots,(m_{t},n_{t})}$ be a non-zero linear sum-rank-metric code, then
\begin{itemize}
  \item[(1)] $d_{1}(C)=d_{sr}(C)$.
  \item[(2)] $d_{r}(C)\leq d_{s}(C)$ for $1\leq r\leq s\leq\dim(C)$.
  \item[(3)] $d_{\dim(C)}(C)\leq n$, $n=n_{1}+\cdots+n_{t}$.
  \item[(4)] $d_{r+m_{1}n_{1}+\cdots+m_{j-1}n_{j-1}+\delta m_{j}}(C)\geq d_{r}(C)+n_{1}+\cdots+n_{j-1}+\delta$ for $j\in[t]$, $r\in[\dim(C)-(m_{1}n_{1}+\cdots+m_{j-1}n_{j-1}+\delta m_{j})]$ and $0\leq\delta\leq n_{j}-1$.
\end{itemize}
\end{lemma}



Then, we can connect the generalized weights of a linear sum-rank-metric code with the minimum sum-rank distance of its dual code.
\begin{lemma}\label{Lemma 4.23}
  Let $C\subseteq\mathbb{F}_{q}^{(m,n_{1}),\ldots,(m,n_{t})}$ be a linear sum-rank-metric code with minimum sum-rank distance $d_{sr}(C)=d$ and dimension $\dim(C)=\alpha m+\beta$, where $n=n_{1}+\cdots+n_{t}$, $0\leq\alpha\leq n-d$ and $0\leq \beta\leq m-1$. If $\dim(C)<m$, then $d_{sr}(C^{\bot})=1$. If $\dim(C)\geq m$, then 
  {\scriptsize $$d_{sr}(C^{\bot})=\begin{cases}
                       \alpha+1, & \mbox{if } n+1-d_{\dim(C)+1-\alpha m}(C)=\alpha,\\
                       \min\{1\leq r\leq\alpha\,:\,n+1-d_{\dim(C)+1-rm}(C)>r\}, & \mbox{otherwise}.
                     \end{cases}$$}
\end{lemma}
\begin{proof}
  If $\dim(C)<m$, then by the Singleton bound, $d_{sr}(C^{\bot})\leq n-\left\lceil\frac{mn-\dim(C)}{m}\right\rceil+1=1$, giving $d_{sr}(C^{\bot})=1$.
  Now we assume that $\dim(C)\geq m$. According to Lemma \ref{Lemma 4.21} (3), (4), 
  {\small \begin{equation*}
   \begin{aligned}
   & d_{sr}(C^{\bot})=d_{1}(C^{\bot})<d_{1+m}(C^{\bot})<\cdots<d_{1+(n-\alpha-1)m}(C^{\bot})\leq n,\\
   & n+1-d_{\dim(C)+1-m}(C)<n+1-d_{\dim(C)+1-2m}(C)<\cdots<n+1-d_{\dim(C)+1-\alpha m}(C).
   \end{aligned}
   \end{equation*}}\\
   Let {\scriptsize \begin{equation*}
   \begin{aligned}
  & \mathcal{W}_{1}(C^{\bot})  =\{d_{1}(C^{\bot}),\,d_{1+m}(C^{\bot}),\ldots,d_{1+(n-\alpha-1)m}(C^{\bot})\}, \\
  & \overline{\mathcal{W}}_{1+\dim(C)}(C)  =\{n+1-d_{\dim(C)+1-m}(C),\,n+1-d_{\dim(C)+1-2m}(C),\ldots,n+1-d_{\dim(C)+1-\alpha m}(C)\}. 
   \end{aligned}
   \end{equation*} }\\
   By \cite[Theorem VI.9]{CGLGMS}, $\mathcal{W}_{1}(C^{\bot})=[n]\setminus \overline{\mathcal{W}}_{1+\dim(C)}(C)$. 
   Hence, we can get the conclusion that 
   {\scriptsize $$d_{sr}(C^{\bot})=\begin{cases}
                       \alpha+1, & \mbox{if } n+1-d_{\dim(C)+1-\alpha m}(C)=\alpha,\\
                       \min\{1\leq r\leq\alpha\,:\,n+1-d_{\dim(C)+1-rm}(C)>r\}, & \mbox{otherwise}.
                     \end{cases}$$}
\end{proof}
Now, we are ready to compute the generalized weights of dually QMSRD codes.
\begin{theorem}\label{Theorem 4.27}
  Let $C\subseteq\mathbb{F}_{q}^{(m,n_{1}),\ldots,(m,n_{t})}$ be a linear sum-rank-metric code with minimum sum-rank distance $d_{sr}(C)=d$ and dimension $\dim(C)=\alpha m+\beta$, where $n=n_{1}+\cdots+n_{t}$, $0\leq\alpha\leq n-d$ and $0\leq \beta\leq m-1$.
  \begin{itemize}
    \item[(1)] If $\dim(C)<m$, then $C$ is a dually QMSRD code if and only if $C$ is a QMSRD code which is equivalent to that $d_{1}(C)=n$.
    \item[(2)] If $\dim(C)\geq m$. Then $C$ is a dually QMSRD code if and only if $d_{1}(C)=n-\alpha$ and $d_{\beta+1}(C)=n+1-\alpha$.
    In addition, if $C$ is a dually QMSRD code, the generalized sum-rank weights of $C$ is given as below:
  \begin{equation*}
   \begin{aligned}
   & d_{1}(C)=\cdots=d_{\beta}(C)=n-\alpha,\\
   & d_{\beta+1+sm}(C)=\cdots=d_{\beta+(s+1)m}(C)=n+1+s-\alpha, \text{ for } 0\leq s\leq\alpha-2,\\
   & d_{\beta+1+(\alpha-1)m}(C)=\cdots=d_{\dim(C)}=n.  
   \end{aligned}
   \end{equation*}
  \end{itemize}
\end{theorem}
\begin{proof}
  \begin{itemize}
    \item[(1)] If $\dim(C)<m$, $$d_{sr}(C^{\bot})\leq n-\left\lceil\frac{mn-\dim(C)}{m}\right\rceil+1=1,$$ 
    
    giving that $C^{\bot}$ is a QMSRD code. Thus, $C$ is a dually QMSRD code if and only if $C$ is a QMSRD code, which implies that $$d_{1}(C)=d=n-\left\lceil\frac{\dim(C)}{m}\right\rceil+1=n.$$
    \item[(2)] If $\dim(C)\geq m$, $C$ is a dually QMSRD code if and only if 
    {\small $$d_{1}(C)=d=n-\left\lceil\frac{\dim(C)}{m}\right\rceil+1=n-\alpha \text{ and } d_{sr}(C^{\bot})=n-\left\lceil\frac{mn-\dim(C)}{m}\right\rceil+1=\alpha+1.$$\\}
    By Lemma \ref{Lemma 4.23}, $d_{sr}(C^{\bot})=\alpha+1$ is equivalent to $n+1-d_{\dim(C)+1-\alpha m}(C)=n+1-d_{\beta+1}(C)=\alpha$, giving $d_{\beta+1}(C)=n+1-\alpha$. Since 
    \begin{equation*}
   {\scriptsize\begin{aligned}
    & n+1-d_{\dim(C)+1-\alpha m}(C)=\alpha \text{ and }\\ & n+1-d_{\dim(C)+1-m}(C)<n+1-d_{\dim(C)+1-2m}(C)< \cdots <n+1-d_{\dim(C)+1-\alpha m}(C),
    \end{aligned}}
   \end{equation*}
   we have $n+1-d_{\dim(C)+1-rm}(C)=r$. Hence, for $0\leq s\leq\alpha-1$,
   $$d_{\beta+1+sm}(C)=d_{\dim(C)+1-(\alpha-s)m}(C)=n+1-(\alpha-s)=n+1+s-\alpha.$$
   In addition, by Lemma \ref{Lemma 4.21} (2) and (3), $n=d_{\beta+1+(\alpha-1)m}(C)\leq d_{\dim(C)}(C)\leq n$, so  
   {\scriptsize$$d_{\beta+1+(\alpha-1)m}(C)=d_{\beta+2+(\alpha-1)m}(C)=\cdots=d_{\beta+(m-1)+(\alpha-1)m}(C)=d_{\dim(C)}(C)=n.$$}\\
   Also, by Lemma \ref{Lemma 4.21} (2) and (4), $$n+1+s-\alpha=d_{\beta+1+sm}(C)\leq d_{\beta+(s+1)m}(C)\leq d_{\dim(C)}-(\alpha-s-1)=n-\alpha+s+1.$$
   Then $$d_{\beta+1+sm}(C)=\cdots=d_{\beta+(s+1)m}(C)=n-\alpha+s+1.$$
   At last, by Lemma \ref{Lemma 4.21} (2) and (4), $$n-\alpha=d_{1}(C)\leq d_{\beta}(C)\leq d_{\beta+m}(C)-1=n-\alpha.$$
   Hence, $$d_{1}(C)=\cdots=d_{\beta}(C)=n-\alpha.$$
  \end{itemize}
\end{proof}
\begin{example}
  Using the theorem above, we compute the generalized sum-rank weights of the dual code of $C$ in Example \ref{Example 3.4}:
  $$d_{1}(C)=2,\quad d_{2}(C)=d_{3}(C)=3,\quad d_{4}(C)=d_{5}(C)=4.$$
\end{example}

\section{Conclusion}
Sum-rank-metric codes are an important kind of codes. Besides many practical applications, the sum-rank metric is also a natural generalization of the Hamming metric and rank metric, thus providing a common theoretical framework for these two well-studied metrics. 
In this paper, we concentrate on a class of good sum-rank-metric codes, QMSRD codes, which have the largest minimum sum-rank distance by the Singleton bound, but are not MSRD codes.
We first study the properties of QMSRD codes. Then we consider their dual codes and based on the fact that not every QMSRD code has a QMSRD dual code, we define dually QMSRD codes. Next, we give the support, rank-list as well as sum-rank distributions of dually QMSRD codes and end up this paper with the generalized sum-rank weights of dually QMSRD codes.

Apart from these contributions, here are several challenging open problems worth considering: 
\begin{itemize}
  \item Construct more QMSRD codes especially dually QMSRD codes.
  \item Compute the density of dually QMSRD codes among QMSRD codes for the fixed field, matrix sizes, block length and dimension.
  \item Derive more properties of QMSRD codes and dually QMSRD codes in $\mathbb{F}_{q}^{(m_{1},n_{1}),\ldots,(m_{t},n_{t})}$ for arbitrary $m_{i}$s.
\end{itemize}

\newpage
\section*{Acknowledgments}
This research is supported by the National Key Research and Development Program of China (Grant No. 2022YFA1005000), the National Natural Science Foundation of China (Grant No. 62371259), the Fundamental Research Funds for the Central Universities of China (Nankai University), and the Nankai Zhide Foundation.

\end{document}